\documentclass{new_tlp}
\usepackage{savesym}
\usepackage{xspace}
\usepackage{amsmath}
\usepackage{amsfonts}

\let\origproof\proof
\usepackage{z-eves}

\let\proof\origproof
\savesymbol{comp}
\savesymbol{defs}
\usepackage{xcolor}
\usepackage{soul}
\usepackage{enumitem}
\usepackage{url}
\usepackage{semantic}
\restoresymbol{scomp}{comp}
\usepackage{graphicx}
\usepackage{comment}
\usepackage{fancyvrb}

\newcommand{\qed}{\hfill$\Box$}

\newcommand{\false}{false}
\newcommand{\defs}{\mathbin{\widehat{=}}}

\newcommand{\setlog}{$\{log\}$\xspace}
\newcommand{\SET}{\{\cdot\}}

\newcommand{\SATSET}{\mathit{SAT}_{\SET}}

\newcommand{\e}{\varnothing}
\renewcommand{\plus}{\mathbin{\scriptstyle\sqcup}}
\newcommand{\Disj}{disj}
\newcommand{\Size}{size}

\renewcommand{\Cup}{un}
\renewcommand{\Cap}{inters}
\newcommand{\Ncup}{nun}
\newcommand{\Id}{id}
\newcommand{\Inv}{inv}
\newcommand{\Comp}{comp}

\newcommand{\TYPES}{\mathsf{Types}_{\SET}}  
\newcommand{\Dec}{dec}
\newcommand{\Etype}{\mathsf{enum}}
\newcommand{\Utype}{\mathsf{sum}}
\newcommand{\Stype}{\mathsf{set}}
\newcommand{\Rtype}{\mathsf{rel}}
\newcommand{\Itype}{\mathsf{int}}
\newcommand{\Ptype}{\mathsf{prod}}
\newcommand{\qm}{?}
\newcommand{\Ut}{\mathcal{B}}
\newcommand{\At}{\mathcal{A}}
\newcommand{\ac}[1]{\textbf{#1}}
\newcommand{\ut}[1]{\textsf{#1}}
\newcommand{\uc}[2]{\ut{#1}?\ac{#2}}
\newcommand{\tjud}[1]{\Gamma\vdash #1}
\newcommand{\gtjud}[2]{\{#1\} \cup \Gamma\vdash #2}

\newtheorem{example}{Example}

\newtheorem{theorem}{Theorem}

\title[Combining Type Checking and Set Constraint Solving]
{Combining Type Checking and Set Constraint Solving to Improve Automated Software Verification}

\author[M. Cristi\'a and G. Rossi]
{MAXIMILIANO CRISTI\'A \\
Universidad Nacional de Rosario and CIFASIS \\
Argentina    \\
E-mail: cristia@cifasis-conicet.gov.ar
\and
GIANFRANCO ROSSI \\
Universit\`a di Parma \\
Italy    \\
E-mail: gianfranco.rossi@unipr.it
}

\jdate{n/a}
\pubyear{n/a}
\pagerange{\pageref{firstpage}--\pageref{lastpage}}
\doi{n/a}

\begin{document}
\label{firstpage}

\maketitle

\begin{abstract}
This technical note shows how we have combined prescriptive type
checking and constraint solving to increase automation during
software verification.
We do so by defining a type system and implementing a
typechecker for \setlog (read `setlog'), a Constraint Logic
Programming (CLP) language and satisfiability solver based on set theory. The
constraint solver is proved to be safe w.r.t. the type system.
Two industrial-strength case studies are
presented where this combination is used with very good results.

Under consideration in Theory and Practice of Logic Programming (TPLP).\end{abstract}


\section{Introduction}

CLP systems can be used for formal software verification. In general, these systems provide untyped
languages. Adding \emph{prescriptive} type systems to CLP languages may help when CLP is used for formal verification because
typecheckers are sound and complete.

Prescriptive type checking catch errors at compile time,
whereas constraint solving can catch errors at runtime. Catching errors at
runtime might not be acceptable in certain contexts---e.g., safety-critical
systems. In this note we show how prescritive type checking and set constraint
solving can be combined in CLP to automatically found errors before the program
is executed. We do so by defining a type system and implementing a typechecker
for \setlog (`setlog') \cite{Dovier00,setlog}.

\setlog is a satisfiability solver implementing decision procedures for several
expressive fragments of the theory of finite sets and finite set relation
algebra.
Several in-depth empirical evaluations provide evidence that \setlog
is able to solve non-trivial problems (\citeANP{DBLP:journals/jar/CristiaR20}
\citeyearNP{DBLP:conf/RelMiCS/CristiaR18,DBLP:journals/jar/CristiaR20,DBLP:journals/jar/CristiaR21a,Cristia2024};
\citeNP{CristiaRossiSEFM13}); in particular as an automated verifier of
security properties (\citeANP{DBLP:journals/jar/CristiaR21}
\citeyearNP{DBLP:journals/jar/CristiaR21,DBLP:journals/jar/CristiaR21b}) \cite{DBLP:journals/jar/CristiaLL23}.
That is, \setlog is able to automatically prove that a program or a model verifies some
complex properties such as the invariance lemmas arising from the verification of non-trivial state machines.
Rooted in CLP and Prolog, \setlog essentially provides an
untyped language. The lack of
types makes it impossible for \setlog to find certain classes of errors.

In this note, we show how \setlog has been enriched with a (prescriptive) type
system. \setlog's type system is an instance of a Hindley-Milner system
including parametric polymorphism for function and predicate symbols. Given
that in \setlog finite sets and set operators are first-class entities of the
language, the type system is defined as to provide support for them. That is,
there are types for sets, binary relations and their elements, and set and
relational operators are typed accordingly. The type system is based on the
type system defined for the Z formal notation \cite{Spivey00}, which in turn is
similar to B's \cite{Abrial00}. Actually, \setlog has been proposed as a
prototyping language for B and Z specifications
\cite{CristiaRossiSEFM13,DBLP:journals/jar/CristiaR21b,Cristia2024}.

Through the combination between type  checking and
CLP features \setlog performs three important tasks: it type checks a program,
it runs the program, and it automatically proves properties of the same
program. As far as we understand this is a novel approach concerning tools for
set-based languages. There are tools where types are used to guarantee some
properties (e.g., Atelier-B \cite{Mentre00}, Rodin
\cite{DBLP:journals/sttt/AbrialBHHMV10} and Zenon Modulo
\cite{DBLP:conf/lpar/BonichonDD07}) but they do not enjoy CLP properties---for
instance, they cannot execute their models. Other tools clearly fit in the CLP
paradigm (e.g., ProB \cite{Leuschel00}), but they cannot prove properties true
of the specification. None of these tools implement decision procedures for set
theory as those implemented in \setlog. At the same time, \setlog provides a
language almost at the same level of abstraction of formal notations such as B and Z.
These features make \setlog to solve real-world problems as shown in Section \ref{typesCLP}.

As a design choice, we allow the typechecker in \setlog to be switched on and
off by users; when switched off, \setlog works as usual. This choice is not
only a matter of backward compatibility but the result of understanding that
typed and untyped formalisms have their own advantages and disadvantages
\cite{DBLP:journals/toplas/LamportP99}.  Hence, being \setlog at the intersection of a specification language, a
programming language and an automated verifier, we think that users can decide
when they need types.

The contents of this note are rather informal and conceptual. All the technical details can be found in an on-line document \cite{DBLP:journals/corr/abs-2205-01713}.

\section{\label{motiv}Combining type checking and constraint solving}

As we have said, \setlog 
provides an untyped language. For example, if variable \verb+X+
represents the current value of a traffic light it should always be the case that \verb+X in {red,yellow,green}+ holds (where \verb+in+ is interpreted as $\in$). 
However, in a \setlog program it is possible
for \verb+X+ to be bound to any value, say \verb+2+. Then, \verb+2 in {red,yellow,green}+ will fail making \setlog to answer \verb+false+ at
runtime---which is potentially dangerous in a safety-critical system. \setlog will also answer \verb+false+ if \verb+X+ is bound to \verb+green+ when \verb+X in {red,yellow}+ holds. Hence, we have two
\verb+false+ answers arising from quite different causes: the first one (\verb+X = 2+)
is clearly a programming error; the second one (\verb+X = green+) is just a
possible, although unsatisfiable, situation. We would like to help the
programmer to easily tell these \verb+false+ answers apart. A (prescriptive)
typechecker can signal the user with a type error if there is the chance of \verb+X+
being bound to something outside of \verb+{red,yellow,green}+;
the constraint solver can determine whether or not \verb+X in {red,yellow}+ is satisfiable by discharging a certain proof obligation. Both
checks are performed before the program is run, thus avoiding runtime errors.

\subsection{\label{library}Typechecking and formally verifying a simple library system}

Consider the following \setlog clause concerning a simple library system: 
\begin{verbatim}
addBook(Books,B,T,Books_) :-
  dom(Books,D) & B nin D & Books_ = {[B,T] / Books}.
\end{verbatim}
where \verb+Books+ is a set of ordered pairs mapping books IDs (BID) onto their titles and \verb+Books_+ is the state of the library after adding
\verb+[B,T]+ to it. In turn, \verb+B+ and \verb+T+ are the book ID and title of a new book being added to the library. Besides, \verb+dom/2+ is a \setlog constraint computing the domain of a function; \verb+nin+ is interpreted as $\notin$; and \verb+{[B,T] / Books}+ is a \setlog set term interpreted as $\{(B,T)\} \cup Books$, where $(B,T)$ is an ordered pair. Finally, \verb+&+ stands for conjunction.

Without types, \verb+T+ can be bound to any term (e.g., an ordered pair or a list) which is clearly not the intention. After adding types to \setlog we can declare the type of \verb+addBook+ as follows:
\begin{verbatim}
dec_p_type(addBook(rel(bid,title),bid,title,rel(bid,title)).
addBook(Books,B,T,Books_) :-
  dom(Books,D) & B nin D & Books_ = {[B,T] / Books}.
\end{verbatim}
That is, the clause is preceded by a \verb+dec_p_type+ fact asserting the type of each
argument\footnote{\texttt{dec\_p\_type} stands for `\textit{dec}lare
\textit{p}redicate \text{type}'.}. As can be seen, \verb+dec_p_type+ takes one argument corresponding to the head
predicate of the clause being declared. In turn, each argument of the
predicate inside \verb+dec_p_type+ is
a type. In this way, \verb+rel(bid,title)+ is the type of
\verb+Books+, \verb+bid+ is the type of \verb+B+, etc.

\verb+rel(bid,title)+ corresponds to the type of the binary relations from \verb+bid+ to \verb+title+, which in turn are \emph{basic types}. Values of basic type \verb+x+ are of the form \verb+x?+$atom$ where $atom$ is any Prolog atom.

When the file containing \verb+addBook/4+ is consulted, the typechecker checks the
type of \verb+addBook/4+.
Later, users can issue queries involving the clauses
declared in the file. In this case, users have to give the types of all the
variables involved in the query. For example:
\begin{verbatim}
dec(NewBooks,rel(bid,title)) & addBook({},bid?b1,title?aleph,NewBooks)
\end{verbatim}
where \verb+dec(NewBooks,rel(bid,title))+ declares \verb+NewBooks+'s type and \verb+{}+ denotes the empty set.
Before executing the query, \setlog calls the typechecker to check whether the
query is type correct or not. In this case, the typechecker uses the
\verb+dec_p_type+ of \verb+addBook/4+, the \verb+dec+ predicates and the arguments in the
query to control that each argument in the query has the right type. 

Types ensure that programs ``are free
from certain kinds of misbehavior'' \cite{DBLP:books/daglib/0005958}. However, types cannot catch all
kinds of errors. Concerning the library system, observe that \verb+Books+ is supposed to be a function but its type declaration states that it is a binary relation. Then, this type declaration cannot
catch the error of putting in \verb+Books+ two ordered pairs sharing the same BID but different titles. 

One may think in a type system where functions are first-class entities. However, in set-based formal notations such as B and Z binary relations are first-class entities whereas functions are defined as a subclass of binary relations. Hence, in this context this problem is approached by
asking the user to discharge \emph{proof obligations} ensuring that, for
instance, a binary relation is actually a function. \setlog follows the same
approach. Actually, \setlog uses its automated proving power to
discharge those proof obligations. Hence, \setlog combines type checking and constraint
solving to ensure program correctness.

In this way, if we want \verb+Books+ to be a function we
can use the \verb+pfun+ constraint in combination with a type declaration:
\verb+dec(Books,rel(bid,title)) & pfun(Books)+. \verb+pfun+ is a \setlog constraint asserting that its argument is a
partial function.

This approach is particularly amenable to work with \emph{invariance lemmas}.
An invariance lemma states that some property, the \emph{invariant}, is
preserved along all program executions.
Indeed, one can propose a typing property that cannot be enforced by the type
system as a program invariant. Afterwards, one proves that the program
preserves that invariant by discharging the corresponding invariance lemma. The
obvious problem with this approach is the  fact that discharging an invariance
lemma requires, in general, a manual proof. However, if the invariance lemma corresponds to a formula
belonging to some decidable theory, it can, in principle, be automatically
discharged. Here is where \setlog constraint solving capabilities come into
play.

\begin{example}\label{e:typeinv}
The type system
ensures that, for instance, \verb+[B,12]+ cannot be added to \verb+Books+, but it cannot
ensure that \verb+Books+ is a function.
On the other hand, \setlog (without types) can automatically prove that \verb+pfun(Books)+ is a
type invariant for that operation but it cannot check that \verb+[B,12]+ cannot be
an element of \verb+Books+ before the program is executed.
In order to prove that
\verb+pfun(Books)+ is a type invariant the following is (automatically) proved to
be unsatisfiable:
\begin{verbatim}
dec([Books,Books_],rel(bid,title)) & dec(B,bid) & dec(T,title) &
pfun(Books)  &  addBook(Books,B,T,Books_)  &  npfun(Books_)
\end{verbatim}
where \verb+npfun+ is interpreted as $\lnot$\verb+pfun+. We know this can be automatically proved because the formula belongs to a
decidable fragment implemented in \setlog (see Section \ref{constraintsolver}). 
\qed
\end{example}

Furthermore, types simplify some proof obligations thus further reducing the
verification effort. The following is an example.

\begin{example}\label{ex:lessproofs}
Consider that in the library system $\texttt{Books} \in \texttt{bid} \pfun \texttt{title}$
should be an invariant of the system. This invariant can be mathematically rewritten as:
$\dom \texttt{Books} \subseteq \texttt{bid} \land \ran \texttt{Books} \subseteq \texttt{title} \land
\texttt{pfun}(\texttt{Books})$. With the type declaration
\verb+dec(Books,rel(bid,title))+ the first two conjuncts are proved by
the typechecker. Then, \verb+pfun(Books)+ is the only property the constraint
solver has to prove.
\qed
\end{example}

As set constraint solving is always less efficient than type checking, the
combination between type checking and constraint solving makes a better use of
the computing resources because type checking finds errors before constraint solving is called.

Once \setlog has been used to verify the program, it can also be used
to run simulations.

\begin{example}\label{ex:addBook}
The predicate \verb+addBook+ can be called as a normal subroutine.  
\begin{verbatim}
{log}=> addBook({},bid?b1,title?the_farm,Books_).
Books_ = {[bid?b1,title?the_farm]}
{log}=> addBook({},bid?b1,title?the_farm,Books1)
        & addBook(Books1,bid?b2,title?houses,Books_).
Books_ = {[bid?b1,title?the_farm],[bid?b2,title?houses]}
{log}=> addBook({},bid?b1,title?the_farm,Books1)
        & addBook(Books1,bid?b1,title?houses,Books_).   % same ID
false
\end{verbatim}
We can also define other operations on the library:
\begin{verbatim}
dec_p_type(title(rel(bid,title),bid,title)).
title(Books,B,T) :- applyTo(Books,B,T).
\end{verbatim}
where \verb+applyTo(F,X,Y)+ is true when $F(X) = Y$. Then we can run the program calling all the operations:
\begin{verbatim}
{log}=> addBook({},bid?b1,title?the_farm,Books1)
        & addBook(Books1,bid?b2,title?houses,Books_)
        & title(Books_,bid?b1,T).
T = title?the_farm
\end{verbatim}
\qed
\end{example}

\section{\label{language}A type system for \setlog}

In this section we describe the type system defined for \setlog. After the introduction of types, every variable, term and predicate is of or has a type\footnote{The type of a predicate is given by the type of each of its arguments. See examples in Section \ref{motiv}.}. 
The type system includes types for a subclass of Prolog atoms,
integer numbers, sum types\footnote{Tagged union, variant, variant record,
choice type, discriminated union, disjoint union, or coproduct.}, ordered pairs
and sets, which can be recursively combined. This recursion enables the
definition of types corresponding to sets of ordered pairs which allow the
correct typing of binary relations. Terms of any type can be used to build set
terms as long as they are all of the same type, no nesting restrictions being
enforced (in particular, membership chains of any finite length can be
modeled). In this and the following section we use a more abstract, math-oriented notation, instead of the \setlog code shown in the previous section and in Section \ref{setlog}. A fully detailed presentation can be found in the on-line document \cite{DBLP:journals/corr/abs-2205-01713}.

The type system is given by the following grammar:
\begin{flalign*}
& \tau ::= \Itype | b | \Utype([\overbrace{C,\dots,C}^{\geq 2}]) | \Ptype(\tau,\tau) | \Stype(\tau) \\
& C ::= a | a(\tau)
\end{flalign*}
where $b \in \Ut$ and $a \in \At$ whereas $\At$ represent the set of Prolog atoms and $\Ut$ is a countable set of \emph{type names}.

Intuitively, $\Itype$ corresponds to the type of integer numbers;
$\Utype([C_1,\dots,C_n])$ with $2 \leq n$ defines a \emph{sum type} given by
all the values that can be built by each constructor $C_i$; $\Ptype(T_1,T_2)$
defines the \emph{product type} given by the Cartesian product $T_1
\cross T_2$; and $\Stype(T)$ defines the \emph{powerset type} of type $T$. Each
constant $b \in \Ut$ represents the \emph{type name} of a \emph{basic type}
interpreted as the set $\{b \qm a | a \in \At\}$. 
An example of a sum type is $\Utype([\ac{nil},\ac{some}(\Itype)])$ representing
the set $\{\ac{nil}\} \cup \{\ac{some}(i) | i \in \num\}$.
We write
$\Rtype(\tau_1,\tau_2)$ as a shortcut for $\Stype(\Ptype(\tau_1,\tau_2))$.
Clearly, $\Rtype$ represents the type of binary relations.
This type system is aligned with those of Z \cite{Spivey00} and B \cite{Abrial00}.

The type of the main \setlog function symbols is defined in Figure
\ref{f:typefunc}, where: 
$dec(v,t)$ is a predicate interpreted as ``variable $v$ is of type $t$''; 
a typing context $\Gamma$ is a set of $\Dec$ predicates; 
$\tjud{s:\tau}$ can be read as `in typing context $\Gamma$, $s$ is of or has type $\tau$';
$v \notin \dom \Gamma$ means that $v$ is not
one of the variables declared in $\Gamma$; 
$\tau = \dots\Utype([a_1,a_2(\tau_2)])\dots$ asserts that
$\Utype([a_1,a_2(\tau_2)])$ is part of the definition of $\tau$; 
and $\{a_1,a_2\} \parallel (\tau \cup \Gamma)$ asserts that $\{a_1,a_2\}$ is disjoint w.r.t. the union of the constructors of any other $\Utype$ used in $\tau$ and $\Gamma$ (except itself). Symbols $\{\}$, $[\cdot,\cdot]$ and $\{\cdot/\cdot\}$ were introduced in Section \ref{motiv}; $int(m,n)$ stands for the set $\{z \in \num \mid m \leq z \leq n\}$;
and $\boxplus$ stands for any integer binary operator.

Rule \textsc{Sum} is given for a
simple case to keep the presentation manageable; it can be easily extended to
any number of constructors. Note that a sum constructor can be used in only one
sum type, and that there must be at least two constructors in a sum type. 

As can be seen, several function symbols are polymorphic. In particular, the
type of a term of the form $\ut{u}?a$ is given by its first argument.
 As terms must have exactly one type, any $a \in \At$ can be used
as the constructor of at most one $\Utype$ in $\Gamma$. The type of a variable
is the type given in its declaration through the $\Dec$ predicate in $\Gamma$.
With this type system a set such as $\{\ac{a}, 1, \{2\}, (5,4)\}$ is ruled out because not all of its elements are of the same type. However, that set can be encoded as $\{\ac{a}, \ac{n}(1), \ac{s}(\{2\}),
\ac{p}((5,4))\}$, where $\ac{a}$, $\ac{n}$, $\ac{s}$ and $\ac{p}$ are
constructors of a $\Utype$ type. If in a $\Utype$ all constructors are nullary
terms then we write $\Etype([a_1,\dots,a_n])$ to emphasize the fact that the
sum is actually defining an enumeration.

\begin{figure}
\begin{equation*}
\begin{array}{cc}
\inference{}{\tjud{\{\}: \Stype(\tau)}} & 
%
\inference[\textsc{Ext}]
    {\tjud{x:\tau} & \tjud{A:\Stype(\tau)}}
    {\tjud{\{x / A\}: \Stype(\tau)}} \\[5mm]
%
\inference[\textsc{Int}]
    {\tjud{m:\Itype} & \tjud{n:\Itype}}
    {\tjud{int([m,n): \Stype(\Itype)}} & 
%
\inference{n \in \num}{\tjud{n: \Itype}} \\[5mm]
%
\inference[\textsc{Ar}]
    {\tjud{m:\Itype} & \tjud{n:\Itype}}
    {\tjud{m \boxplus n: \Itype}} & 
%
\inference[\textsc{Prod}]
    {\tjud{x:\tau_1} & \tjud{y:\tau_2}}
    {\tjud{[x,y]: \Ptype(\tau_1,\tau_2)}} \\[5mm]
%
\inference
  {b \in \Ut & a \in \At_0}
  {\tjud{b \qm a: b}} & 
%
\inference
  {v \in Var & \tau \in \TYPES & v \notin \dom \Gamma}
  {\gtjud{\Dec(v,\tau)}{v:\tau}}
\end{array}
\end{equation*}
\begin{gather*}
%
\inference[\textsc{Sum}]
    {\{a_1,a_2\} \subset \At &
     \tau = \dots\Utype([a_1,a_2(\tau_2)])\dots &
     \{a_1,a_2\} \parallel (\tau \cup \Gamma) &
     x:\tau_2}
    {\gtjud{\Dec(v,\tau)}{a_1,a_2(x): \Utype([a_1,a_2(\tau_2)])}}
\end{gather*}
\caption{\label{f:typefunc}Type rules for function symbols}
\end{figure}

\begin{example}
The following are some examples on how terms are typed.
\begin{align*}
& \gtjud{\Dec(k,\Itype)}{k + 1: \Itype} \\
& \gtjud{\Dec(A,\Stype(\Itype))}{\{4,9 / A\}: \Stype(\Itype)} \\
& \{\Dec(m,\Itype), \Dec(x,\Utype([\ac{a},\ac{e}(\Itype)])\} \vdash
  (\{m\},\ac{e}(m+1)):
    \Ptype(\Stype(\Itype),\Utype([\ac{a},\ac{e}(\Itype)])) \\
& \tjud{\{\uc{u}{a},\uc{u}{aa},\uc{u}{ab}\}:\Stype(\ut{u})} \\
& \{\uc{u}{a},\uc{v}{a}\}: \text{cannot be typed because \uc{u}{a} has type \ut{u} and \uc{v}{a} has type \ut{v}} \\
& \ut{u}?10: \text{cannot be typed because $10 \notin \At$} \\
& \gtjud{\Dec(x,\Utype([\ac{a},\ac{b},\ac{a}(\ut{t})]))}
        {\dots} \\
&\qquad\text{ is an ill-formed typing context because \ac{a} appears twice} \\
& \gtjud{\Dec(x,\Ptype(\Etype([\ac{a},\ac{b}]),
                       \Utype([\ac{q},\ac{a}(\ut{t})]))}
        {\dots} \\
&\qquad\text{ is an ill-formed typing context because \ac{a} appears in two different $\Utype$} \\
& \gtjud{\Dec(x,\Etype([\ac{a,b}])),
                        \Dec(y,\Etype([\ac{q},\ac{a}]))}
        {\dots} \\
&\qquad\text{ is an ill-formed typing context because \ac{a} appears in two
different $\Etype$}
\end{align*}
In the last three examples the problem is that the terms \ac{a}, \ac{b} and
\ac{q} cannot be typed because no type rule can be applied to infer their
types. 
\qed
\end{example}

The type of each primitive constraint available in \setlog is given in Figure \ref{f:typepred}.
If $\tau_1,\dots,\tau_n$ are types, a type judgment such as $\tjud{\pi(\tau_1,\dots,\tau_n)}$ can be read as `predicate $\pi$ is correctly typed'.
Constraints are interpreted as follows: $x \mathbin{neq} y$ as $x \neq y$; $x \mathbin{in} A$ as $x \in A$; $x \mathbin{nin} A$ as $x \notin A$; $\Cup(A,B,C)$ is interpreted as $C = A \cup B$; $\Disj(A,B)$ as $A \cap B = \e$; $\Size(A,m)$ as $\lvert A \rvert = m$; $\Id(A,R)$ as $R = \{(x,x) \mid x \in A\}$, i.e., the identity relation on set $A$; $\Inv(R,S)$ as $S = \{(y,x) \mid (x,y) \in R\}$, i.e., the converse of a binary relation; and $\Comp(R,S,T)$ as $T = \{(x,z) \mid \exists y ((x,y) \in R \land (y,z) \in S)\}$, i.e., composition of binary relations.
Again, several constraints are polymorphic. Besides, it is worth to be noticed
that in untyped formalisms a constraint such as $1 \neq (2,4)$ would usually be
deemed as a true proposition while in \setlog such a constraint is ill-typed
and thus rejected.

\begin{figure}
\begin{equation*}
\begin{array}{ccc}
\inference[\textsc{Eq}]
  {\tjud{x:\tau} & \tjud{y:\tau}}
  {\tjud{x = y, x \mathbin{neq} y}} & 
%
\hspace{2cm}
\inference[\textsc{Mem}]
  {\tjud{x:\tau} & \tjud{A:\Stype(\tau)}}
  {\tjud{x \mathbin{in} A, x \mathbin{nin} A}} 
\end{array}
\end{equation*}
\begin{gather*}
\inference[\textsc{Un}]
  {\tjud{A:\Stype(\tau)} & \tjud{B:\Stype(\tau)} & \tjud{C:\Stype(\tau)}}
  {\tjud{\Cup(A,B,C)}} 
\end{gather*}
\[
\begin{array}{cc}
\inference
  {\tjud{A:\Stype(\tau)} & \tjud{B:\Stype(\tau)}}
  {\tjud{\Disj(A,B)}} & 
%
\inference[\textsc{Sz}]
  {\tjud{A:\Stype(\tau)} & \tjud{m:\Itype}}
  {\tjud{\Size(A,m)}} \\[6mm]
%
\inference[\textsc{Id}]
  {\tjud{A:\Stype(\tau)} & \tjud{R:\Rtype(\tau,\tau)}}
  {\tjud{\Id(A,R)}} & 
%
\inference
  {\tjud{R:\Rtype(\tau_1,\tau_2)} & \tjud{S:\Rtype(\tau_2,\tau_1)}}
  {\tjud{\Inv(R,S)}} 
\end{array}
\]
\begin{gather*}
\inference[\textsc{Rc}]
  {\tjud{R:\Rtype(\tau_1,\tau_2)} & \tjud{S:\Rtype(\tau_2,\tau_3)} & \tjud{T:\Rtype(\tau_1,\tau_3)}}
  {\tjud{\Comp(R,S,T)}} \\[1mm]
%
\inference[\textsc{Leq}]
  {\tjud{m:\Itype} & \tjud{n:\Itype}}
  {\tjud{m \leq n}}
\end{gather*}
\caption{\label{f:typepred}Type rules for predicate symbols in $\Pi_C$}
\end{figure}

In \setlog,  formulas are built in the usual way by connecting constraints (i.e., predefined predicates) by means of conjunction ($\land$), disjunction ($\lor$) and negation ($\lnot$). Formulas can also contain user-defined predicates (\texttt{addBook/4} in Section \ref{motiv} is an example).
A well-typed formula is a formula where all its constraints are typed according to the rules given in Figure \ref{f:typepred} and user-defined predicates are typed by means of \verb+dec_p_type+ declarations. 

\begin{example}\label{ex:formulas}
The following are well-typed \setlog formulas:
\begin{align}
& \Id(\{x / A\},R) \land \Id(R,A) \land \Dec(x,\ut{t}) \land \Dec(A, \Stype(\ut{t})) \land \Dec(R,\Rtype(\ut{t},\ut{t})) \\
& \Dec([A,B,C],\Stype(\Itype)) \land \Dec([n,k],\Itype) \land \\
& \qquad \Cup(A,B,C) \land n+k > 5 \land \Size(C,n) \land B \neq \{\} \notag \\
& \Dec(y,\ut{t}) \land \Dec(R,\Rtype(\ut{t},\Itype)) \land \Dec([S,T],\Rtype(\Itype,\ut{t})) \land \Dec(x,\Ptype(\Itype,\ut{t})) \land \\
& \qquad [y,5] \in R \land [2,\uc{t}{b}] \notin S \land \Inv(R,S) \land S = \{x / T\} \notag
\end{align}
On the contrary, $\Dec(x,\ut{t}) \land \Dec(y,\Itype) \land x \mathbin{neq} y$ is not well-typed because $x \mathbin{neq} y$ cannot be typed as $neq$ expects arguments of the same type (rule \textsc{Eq} in Figure \ref{f:typepred}). 
\qed
\end{example}

\paragraph{Expressiveness.}
As we have said, Figure \ref{f:typepred} presents the type of each \setlog primitive constraint. These constraints are used to define a number of integer, set and relational constraints by means of suitable formulas.
\citeN{Dovier00} proved that $\Cup$, $\in$ and $\Disj$ are enough to define constraints implementing the set operators
$\cap$, $\subseteq$ and $\setminus$. For example, $A \subseteq B$ can be
defined by the formula $\Cup(A,B,B)$. In turn, these constraints plus $\Id$, $\Inv$ and $\Comp$ are as expressive as the class of finite set relation algebras
\cite{DBLP:journals/jar/CristiaR20}. Within this class of algebras it is
possible to define many relational operators such as domain, range, domain
restriction, relational image, etc.  
Finally, by adding the $\Size$ constraint and integer intervals, it is possible to define  operators such as the minimum of a set, the
successor function on a set, partition of a set w.r.t. a given number, etc.
\cite{DBLP:journals/tocl/CristiaR24}.

The negated versions of set, relational and integer operators can be introduced in the same way \cite{Dovier00} (\citeANP{DBLP:journals/jar/CristiaR20}
\citeyearNP{DBLP:journals/jar/CristiaR20,DBLP:journals/tplp/CristiaR23,DBLP:journals/tocl/CristiaR24}). 
For example, $\lnot(A \cup B = C)$ is introduced as:
\begin{equation}\label{e:nun}
\Ncup(A,B,C) \defs
     (n \mathbin{in} C \land n \mathbin{nin} A \land n \mathbin{nin} B)
     \lor (n \mathbin{in} A \land n \mathbin{nin} C)
     \lor (n \mathbin{in} B \land n \mathbin{nin} C)
\end{equation}

The combination between sum and product types permit to encode arbitrary
compound terms. For instance, a term of the form $\ac{p}(2,\{4\})$ can be
encoded as $\ac{p}([2,\{4\}])$ whose type is
$\Utype([\ac{p}(\Ptype(\Itype,\Stype(\Itype))), \dots])$.


\section{\label{satcard}Type safety}

The definition of a type system for a CLP language entails
to prove that the operational semantics of the language
preserves the types of variables and terms as it processes any well-typed
formula. This is called \emph{type soundness} or \emph{type safety}
\cite[Chapter 4]{DBLP:books/cu/Ha2016}. The operational semantics of a CLP
language is given primarily by its constraint solving
procedure.

\subsection{\label{constraintsolver}The \setlog constraint solver}

The constraint solver for \setlog with types is the same solver for \setlog without types\footnote{From now on, we will talk of \setlog as both the constraint language and its solver whenever is clear from context.}. In \setlog with types,
the solver is simply run once the type checking phase has finished successfully
and ignores $\Dec$ predicates. The solver for \setlog without types has been thoroughly studied elsewhere
(\citeANP{Dovier00} \citeyearNP{Dovier00} and
\citeANP{DBLP:journals/jar/CristiaR20}
\citeyearNP{DBLP:journals/jar/CristiaR20,DBLP:journals/tplp/CristiaR23,DBLP:journals/tocl/CristiaR24}).
Several results on the decidability of the satisfiability problem for these
languages have been put forward and several empirical studies showing the
practical capabilities of \setlog have
also been reported (\citeANP{DBLP:journals/jar/CristiaR21b}
\citeyearNP{DBLP:journals/jar/CristiaR21,DBLP:journals/jar/CristiaR21b};
\citeNP{CristiaRossiSEFM13}). For this reason here we give only an overview of the \setlog solver.

The \setlog solver is essentially a rewriting system whose core is a collection of
specialized rewriting procedures. Each rewriting procedure applies a few
non-deterministic rewrite rules which reduce the syntactic complexity of
\setlog constraints of one kind. \setlog takes as input a formula,
$\Phi$. In each iteration, \setlog rewrites $\Phi$ into a new formula, called
$\Phi'$, which becomes the input formula for the next iteration.

As is shown in the aforementioned papers, there are three possible outcomes
when \setlog is applied to $\Phi$:
\begin{enumerate}
\item \setlog returns $\false$ meaning that $\Phi$ is unsatisfiable.
\item \setlog cannot rewrite $\Phi'$ anymore and it is not $\false$,
so $\Phi'$ is returned as it is. Since \setlog can
open a number of non-deterministic choices, many such $\Phi'$ can be returned
to the user. The disjunction of all these $\Phi'$ is equivalent to $\Phi$. The
constraints making these formulas are of a particular kind called
\emph{irreducible} (\citeANP{Dovier00} \citeyearNP{Dovier00} and
\citeANP{DBLP:journals/jar/CristiaR20}
\citeyearNP{DBLP:journals/jar/CristiaR20,DBLP:journals/tplp/CristiaR23,DBLP:journals/tocl/CristiaR24}).
\item If $\Phi$ belongs to an undecidable fragment of the language supported by \setlog,
then the above steps will not terminate.
\end{enumerate}

As we have said, the core of \setlog is a collection of rewriting procedures
which in turn contain a collection of rewrite rules. There are about 60 of such
rules which can be found online \cite{calculusBR}. Figure \ref{f:rules} lists
some representative rewrite rules for the reader to have an idea of what they
look like. As can be seen, each rule has the form: $\phi \longrightarrow \Phi_1
\lor \dots \lor \Phi_n$, where $\phi$ is a \setlog constraint and $\Phi_i$, $i
\geq 1$, are \setlog formulas. Besides, each rule is applied depending on some
syntactic conditions on the constraint arguments. For example, rule
\eqref{un:ext1} applies only when the first argument is an extensional set and
the third is a variable (noted as $\dot{B}$); the second argument can be any
set term. On the right-hand side of the rules, $n,n_i,N,N_i$ represent new
variables.

\begin{figure}
\hrule\vspace{3mm}
\begin{flalign}
\quad\quad & \{x  / A\} = \{y / B\} \longrightarrow & \notag  \\
  & \qquad x = y \land A = B & \notag \\
  & \qquad \lor x = y \land \{x / A\} = B & \label{eq:ext} \\
  & \qquad \lor x = y \land A = \{y / B\} & \notag \\
  & \qquad \lor A = \{y / N\} \land \{x / N\} = B & \notag \\[2mm]
& \Cup(\{x / C\}, A, \dot{B}) \longrightarrow & \notag \\
  & \qquad  \{x / C\} = \{x / N_1\}
      \land x \mathbin{nin} N_1 \land B = \{x / N\} & \label{un:ext1} \\
  & \qquad \land (x \mathbin{nin} A \land \Cup(N_1, A, N) & \notag \\
  & \qquad {}\qquad\lor A = \{x / N_2\}
                 \land x \mathbin{nin} N_2 \land \Cup(N_1, N_2,N)) & \notag \\[2mm]
& \Size(\{x \plus A\},m) \longrightarrow & \notag \\
   & \qquad x \mathbin{nin} A \land m = 1 + n \land \Size(A,n) \land 0 \leq n & \label{size:ext} \\
   & \qquad \lor A = \{x / N\} \land x \mathbin{nin} N \land \Size(N,m)
      & \notag 
%
\end{flalign}
 \hrule
\caption{\label{f:rules}Some rewrite rules implemented by $\SATSET$}
\end{figure}

Rule \eqref{eq:ext} is the main rule of set unification \cite{Dovier2006}, a concept at the base of \setlog.
It states when two
non-empty, non-variable sets are equal by non-deterministically and recursively
computing four cases. As an example, by applying rule
\eqref{eq:ext} to $\{1\} = \{1,1\}$ we get: ($1 = 1 \land \{\} = \{1\}) \lor (1 =
1 \land \{1\} = \{1\}) \lor (1 = 1 \land  \{\} = \{1,1\}) \lor (\{\} = \{1 / N\} \land  \{1 / N\} = \{1\})$, which turns out to be true
(due to the second disjunct).

Rule \eqref{un:ext1} is one of the main rules for $\Cup$ constraints.
Observe that this rule is based on set unification, as well. It computes two cases: $x$
does not belong to $A$ and $x$ belongs to $A$ (in which case $A$ is of the form
$\{x / N_2\}$ for some set $N_2$).  In the latter case $x
\mathbin{nin} N_2$ prevents \setlog from generating infinite terms
denoting the same set.

Rule \eqref{size:ext}
computes the size of any extensional set by counting the elements that belong
to it while taking care of avoiding duplicates. This means that, for instance,
the first non-deterministic choice for a formula such as
$\Size(\{1,2,3,1,4\},m)$ will be:
\[
1 \mathbin{nin} \{2,3,1,4\}
  \land m = 1 + n \land \Size(\{2,3,1,4\},n) \land 0 \leq n
\]
which will eventually lead to a failure due to the presence of  $1 \mathbin{nin} \{2,3,1,4\}$.

\subsection{Proving type safety}


The theorem stating type safety for \setlog follows the guidelines
set forth by \citeN[Chapter 6]{DBLP:books/cu/Ha2016}. However, due to the
peculiarities of CLP we have to introduce some modifications to
the theorem. Indeed, most of the foundational work on type systems has been
done in a framework where programs are functions
\cite{DBLP:journals/jcss/Milner78,DBLP:books/daglib/0000395,DBLP:journals/iandc/WrightF94,DBLP:books/daglib/0005958,DBLP:books/daglib/0032840},
so it is natural to adapt some of its concepts and results to CLP.

Intuitively, Harper's theorem 
states that: the type of any given term remains the same during the execution
of the program; and if a term is well-typed, either it is a value or it can be
further rewritten by the system. The first property is called
\emph{preservation} (or \emph{subject reduction}) and the second is called
\emph{progress}.
This theorem is adapted to our context as follows.

\begin{theorem}[\setlog is type safe]\label{t:soundness}
Consider any \setlog constraint $\pi$. Let $t_1:\tau_1, \dots, t_k:\tau_k$ be such that  $\pi(t_1,\dots,t_k)$ is a well-typed constraint.
Then ($\longrightarrow$ denotes any \setlog rewrite rule):
\begin{enumerate}[leftmargin=1cm,label=(\Roman*)]
\item\label{i:pre} If $\pi(t_1,\dots,t_k) \longrightarrow \Phi$, then there
exist types $\tau_1', \dots, \tau_m'$ ($0 \leq m$) such that:
\[
\Dec(n_1,\tau_1') \land \dots \land \Dec(n_m,\tau_m') \land \Phi
\]
is a well-typed formula, were $n_1,\dots,n_m$ are the new variables of $\Phi$.
\item\label{i:pro} $\pi(t_1,\dots,t_k)$ is irreducible, or there exists $\Phi$
such that $\pi(t_1,\dots,t_k) \longrightarrow \Phi$
\qed
\end{enumerate}
\end{theorem}

The proof can be found in the on-line document \cite{DBLP:journals/corr/abs-2205-01713}. As Theorem \ref{t:soundness} suggests, the proof entails to prove that each of the 60 rewrite rules present in \setlog are type safe.

\section{\label{setlog}Implementation}

A typechecker for \setlog has been implemented in Prolog. The program comprises about 1.1 KLOC of which 200 are devoted to error
printing. It uses only basic, standard Prolog predicates. The typechecker can
be downloaded from the \setlog web site \cite[file
\texttt{setlog\_tc.pl}]{setlog}. As we have pointed out in Section \ref{motiv},
the typechecker can be activated and deactivated by users at their will. If the
typechecker is not active $\Dec$ predicates are ignored and \setlog works as
usual.

The typechecker is implemented in four
phases (see Figure \ref{f:phases}): \textsc{phase 1}, the typechecker analyzes
the $\Dec$ predicates in the formula; \textsc{phase 2}, 
the logical structure of the formula is analyzed; \textsc{phase 3}, once we have a conjunction
of atomic constraints the typechecker goes to check the type of each atomic
constraint using rules in Figure \ref{f:typepred}; and
\textsc{phase 4}, the type of each argument, of a given constraint, is checked
using rules in Figure \ref{f:typefunc}.

\begin{figure}
\centering
\includegraphics[scale=.7]{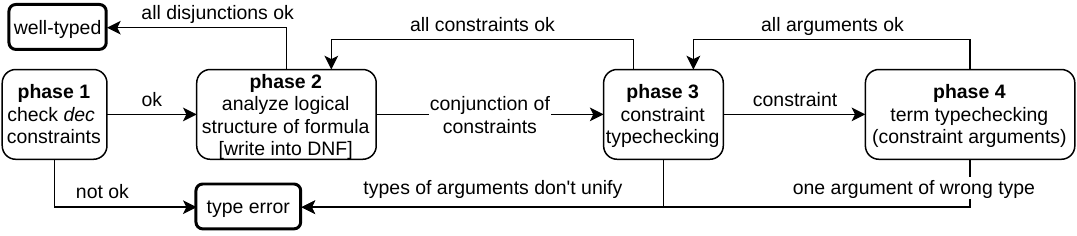}
\caption{\label{f:phases}Graphical depiction of the four phases
implemented by the typechecker}
\end{figure}

Going deeply into implementation details, in \textsc{phase 1}, $\Dec$
predicates that are found correct by the typechecker are turned into `type'
facts of the form \verb+type(Var,type)+, asserting that the type of the variable
is the one declared by the user. 

Phases 3 and 4 are solved as unification problems. For example, the following
clause typechecks $\Id$ constraints (i.e., \textsc{phase 3}):
\begin{verbatim}
typecheck_constraint(id(X,Y)) :- !,
  typecheck_term(X,Tx),
  typecheck_term(Y,Ty),
  (Tx = set(T), Ty = set(prod(T,T)),!
   ;
   print_type_error(id(X,Y),Tx,Ty)
  ).
\end{verbatim}
As can be seen, if the inferred types of \verb+X+ and \verb+Y+ cannot be
unified in terms of a common type \verb+T+ (as indicated in rule \textsc{Id} of
Figure \ref{f:typepred}), a type error is issued.

\subsection{Introducing parametric polymorphism}

Consider the following clause:
\begin{verbatim}
applyTo(F,X,Y) :- F = {[X,Y] / G} & [X,Y] nin G & comp({[X,X]},G,{}).
\end{verbatim}
It states minimum conditions for applying binary relation \verb+F+ to a given point
\verb+X+ whose image is \verb+Y+. \verb+applyTo+ is meant to be a \emph{polymorphic}
definition in the sense that \verb+F+ is expected to be a binary relation of
\emph{any} \verb+rel+ type, \verb+X+ is expected to be of the domain type of \verb+F+ and
\verb+Y+ of the range type. This is usually called \emph{parametric polymorphism}.

Clauses that are meant to be polymorphic must be preceded by a \verb+dec_pp_type+
fact asserting the type of each argument\footnote{\texttt{dec\_pp\_type} stands for
`\textit{dec}lare \textit{p}olymorphic \textit{p}redicate \text{type}'.}. For
example, the typed version of \verb+applyTo+ is the following:
\begin{verbatim}
dec_pp_type(applyTo(rel(T,U),T,U)).
applyTo(F,X,Y) :- F = {[X,Y] / G} & [X,Y] nin G & comp({[X,X]},G,{}).
\end{verbatim}
Differently from \verb+dec_p_type+ declarations, the types written in a
\verb+dec_pp_type+ declaration can contain variables (e.g., \verb+T+). 

Users can query polymorphic clauses such as \verb+applyTo+  by giving types unifying
with those declared in the corresponding \verb+dec_pp_type+ declarations.

\subsection{Dealing with finite types}

When working in typechecking mode, the following goal is unsatisfiable:
\begin{Verbatim}[commandchars=\\\$\$]
dec(Z,enum([t,f])) & Z neq t & Z neq f.
\end{Verbatim}
because the only two values \verb@Z@ can take are exactly \verb@t@ and \verb@f@. In this case, \setlog automatically transforms that goal into:
\begin{Verbatim}[commandchars=\\\$\$]
Z in {t,f} & Z neq t & Z neq f.
\end{Verbatim}
However, if the typechecker is not active, the first given goal is found to be satisfiable because the \verb@dec@ predicate is ignored and \verb@Z@ can take any value beyond \verb@t@ and \verb@f@.

As another example, consider the following goal:
\begin{Verbatim}[commandchars=\\\$\$]
dec(F,rel(enum([t,f]),int)) & dec([X1,X2,X3],[enum([t,f]),int]) & 
pfun(F) & F = {X1,X2,X3} & X1 neq X2 & X1 neq X3 & X2 neq X3.
\end{Verbatim}
As \verb@F@ is a partial function and given the \verb@neq@ constraints, the first components of \verb@X1@, \verb@X2@ and \verb@X3@ must be different from each other. At the same time, these first components have type \verb@enum([t,f])@. So at least two of these first components must have the same value. Consequently the goal is unsatisfiable. As with the first goal, when working in typechecking mode, \setlog identifies this situation and automatically conjoin suitable membership constraints to make type information available to the constraint solver. 

These situations arise when the formula involves variables whose types are finite and entails, in a way or another, ``too many'' \verb@neq@ constraints. In these situations the constraint solver retrieves the type information generated by the typechecker transforming it into suitable membership constraints. This interplay between these two components is crucial to the correctness of \setlog when working in typechecking mode.

\subsection{\label{typesCLP}Case studies}

The combination between typechecking and constraint solving as presented in this note has been applied to some case studies. Here we briefly present two of them. Both case studies are based on problems,
specifications and verification conditions proposed by others, thus reducing a
possible bias towards our method. We use these case studies as benchmarks to
empirically evaluate \setlog. 
Each case study shows a \setlog program taking the form of a state machine.
Transitions of these state machines are encoded as \setlog predicates---e.g., as
\verb+addBook+. In turn, these predicates constitute an executable API from which
the program can be run---as  shown in Example \ref{ex:addBook}. At the same
time, these \setlog programs behave as specifications over which the
verification conditions are automatically discharged---as shown in Example
\ref{e:typeinv}.


\paragraph{The Landing Gear System.}

The first case study is based on the Landing Gear System (LGS) problem
\cite{DBLP:conf/asm/BoniolW14}. \citeANP{DBLP:conf/asm/MammarL14}
(\citeyearNP{DBLP:conf/asm/MammarL14,DBLP:journals/sttt/MammarL17}) developed
an Event-B \cite{Abrial:2010:MES:1855020} specification formally verified using
Rodin \cite{DBLP:journals/sttt/AbrialBHHMV10}, ProB \cite{Leuschel00} and
AnimB\footnote{\url{http://www.animb.org}}. They were able to automatically
discharge 72\% of the proof obligations by calling Rodin.
Hence, we encoded in \setlog the entire Event-B specification of the LGS
and used \setlog to automatically discharge 100\% of the proof obligations
generated by the Rodin tool in roughly 290 s \cite{Cristia2024}. The \setlog encoding of the
LGS\footnote{\url{https://www.clpset.unipr.it/SETLOG/APPLICATIONS/lgs.zip}}
comprises 7.8 KLOC plus 465 proof obligations. In order to discharge all those
proof obligations we used the combination between types and CLP as discussed
above. Many proof obligations could be avoided while many others were simpler,
due to the combined work between the typechecker and the constraint solver---as shown in Example \ref{ex:lessproofs}; see also \cite[Section 4.5]{Cristia2024}.
\emph{The net result is an automatically verified \setlog prototype of the
LGS.}

\paragraph{Android Permission System.}

A model of the Android
permission system has been developed and certified in Coq
\cite{DBLP:conf/ictac/BetarteCLR15,DBLP:journals/cuza/BetarteCLR16,DBLP:journals/cleiej/LunaBCSCG18,DBLP:conf/types/Luca020}.
As with the previous case study, we translated
the Coq model into \setlog and used it to run all the verification conditions proposed in Coq\footnote{\setlog code of Android 10's permission system:
\url{http://www.clpset.unipr.it/SETLOG/APPLICATIONS/android.zip}} \cite{DBLP:journals/jar/CristiaLL23}.   \setlog is able to automatically prove 24 of the 27 properties ($\approx 90\%$) in approximately 21 m. The 3 properties that cannot be proved by \setlog require 500 LOC of Coq proof commands ($\approx 2\%$ of all proof commands). That is, \setlog automatically proves 90\% of the properties covering 98\% of the proof effort in terms of proof commands. As can be seen, the gain in terms of human effort is considerable. The type system proposed for \setlog can encode all the types used in the Coq model. This is important given that Coq is prized by its powerful type system. Then, concerning the Android model, our typechecker can prove the same type properties that are proved by Coq, whereas the constraint solver is able to automatically prove most of the properties that are manually proved in Coq.

\section{\label{relwork}Discussion and related work}

The
inclusion of prescriptive type systems in logic programming\footnote{The technical report by   \citeN{cirstea:inria-00099859} provides an introduction and survey about types in logic programs. There is also a comprehensive book on the matter edited by
\citeN{DBLP:books/mit/pfenning92/P1992}.} can be traced back to the seminal
work of \citeN{Mycroft1984}; later on \citeN{Lakshman1991} define the formal semantics for that
type system. At some point in time, these ideas started to be part of different
logic programming languages
\cite{DBLP:books/daglib/0095081,Schrijvers2008}, and even they made through all
the way to functional logic programming languages
\cite{Somogyi1996,Hanus2013}.  \citeN{Schrijvers2008} propose types to be optional in
SWI-Prolog and Yap, as we do for \setlog.

Descriptive type systems provide an approximation
of the semantics of a given program,
usually, as a set of terms greater than (cf. set inclusion) the one provided by the semantics. Answers
of a descriptive typechecker are necessarily approximate.
These type systems \cite{Zobel1987,Bruynooghe1988,Fruehwirth1991,Dart1992,%
Yardeni1992,Heintze1992,Gallagher1994,DBLP:conf/lopstr/BarbosaFC21%
} have been used in the context of abstract interpretation and
program analysis of logic programs
\cite{Vaucheret2002,Hermenegildo2005,Pietrzak2008}. 

It is worth to be noted that some other works on program approximation,  e.g.,
\cite{Heintze1990,Talbot1997}, use a technique called \emph{set-based
analysis}. This is not to be confused with our approach. We use sets and binary
relations as the main data structures for programming and specification, while
in these other works sets are used to analyze general logic programs. In other
words, we use sets and binary relations to \emph{write} programs, they use sets
to \emph{analyze} them.
Descriptive type systems were also applied to CLP, e.g., as a means to find certain classes of errors in programs
\cite{Drabent2002}. 

\citeN{DBLP:journals/tplp/FagesC01} study a prescriptive type system for CLP
programs that is independent from any constraint domain. The authors prove that
their type system verifies subject reduction (i.e., type preservation) w.r.t.
the abstract execution model of CLP (cf. accumulation of constraints), and
w.r.t. an execution model of CLP with substitutions. As far as we understand,
Fages and Coquery do not prove progress---cf. Theorem \ref{t:soundness}. Clearly, Fages and
Coquery's preservation result reliefs us of proving that \setlog's  CLP engine
preserves types because it is just an implementation of the general CLP scheme
addressed by these authors. On the other hand, Fages and Coquery's result
depends on the well-typedness of each rewrite rule of the execution model. This
is exactly what we prove in Theorem \ref{t:soundness}: that the execution model
 of our CLP language verifies
type safety---i.e., preservation and progress. However, instead of typing each
rewrite rule we prove that they preserve types. We are not aware of previous
works in the field of CLP featuring a formulation of type safety  like that of
Theorem \ref{t:soundness}. 

As we have said, \setlog's type system is inspired in the type system  of Z and
B. Borrowing ideas from these notations is quite natural as \setlog is based on
a set theory similar to those underlying Z and B. In Z and B type inference, in
the sense of deducing the type of some variable from the terms where it
participates in, is not allowed; all variables must be declared. Subtyping is
also nonexistent in Z and B. So we do in \setlog. Differently from Z and B, in
\setlog elements of basic types have a known form (e.g.,
$\uc{b}{a}$). This is useful when \setlog is used as a programming language
because users can give values to input variables of basic types. Obviously,
tools for Z and B implement typecheckers similar to ours.
\citeN{DBLP:journals/corr/abs-2008-02933} shows snippets of the Prolog
implementation of a B typechecker.

Finally, we use the combination between type checking and constraint solving in
a different way compared to works in logic programming. For example,
\citeN{Drabent2002} or \citeN{Pietrzak2008} mentioned above, solve a system of
constraints (sometimes based on set-analysis) to find out whether or not the
program verifies some properties given by means of types. In \setlog, type
checking is used to rule out wrong programs or specifications, and constraint
solving is used as the mechanism of a sort of automated theorem prover. In
\setlog, programs, specifications and properties are set formulas as in Z and
B.

\section{\label{concl}Conclusions}

We have defined a type system and a typechecker for the CLP tool \setlog. Type
checking can be activated or not according to the users needs. The type system
is based on the type systems of formal notations such as B and Z. We have
proved that the operational semantics of \setlog is type safe by adapting the
type safety theorem of functional programming to the CLP context. Although the typechecker and the constraint solver are mostly  independent from each other, they work together when formulas include finite types in order to ensure soundness.
It is our
understanding that the results presented in this paper show that the
combination between a type system and CLP might be a practical approach to
software verification. The case studies based on the LGS problem and the
Android permission system provide empirical evidence about this claim.


As future work we want to study how the introduction of types may help with the
problem of computing the negation of formulas. For example, the negation of
$p(x) \text{ :- } x = (\ac{a},y)$ yields $\forall
y(x \neq (\ac{a},y))$, which lies outside the decision procedures implemented in \setlog. However, if types are brought into the game, the negation of $p(x)$ can be turned into a formula that \setlog can solve. In effect, $x = (\ac{a},y)$ is type-correct
iff $\ac{a}$ is part of a sum type, say $\Utype([\ac{a},\ac{b}(\Itype)])$, $y$
is of some type $\tau_y$, and $x$ is of type
$\Ptype(\Utype([\ac{a},\ac{b}(\Itype)]),\tau_y)$. Then, if $x$ is different
from $(\ac{a},y)$ for \emph{all} $y$, it can only be equal to $(\ac{b}(z),w)$
 for \emph{some} $z$ and \emph{some} $w$, because
otherwise it would lie outside its type. Therefore, $\forall y(x \neq
(\ac{a},y))$ can be turned into $\exists z,w(x = (\ac{b}(z),w))$. \setlog is able to deal with formulas including the latter, meaning that it would be able to deal with formulas including the negation of $p(x)$---something that, without types, could not be achieved. Then, we will work in
finding the class of formulas whose negation can be safely computed when type
information is available. This would be another area where the interplay between the typechecker and the constraint solver produces good results.

\def	\lover		{\mathbin{{\dres} \llap{+\hspace{-3pt}}\;}}
\newcommand{\Ftype}{\mathsf{fun}}

We think that the combination between type checking and constraint solving can
be improved through the introduction of subtypes. For example, if $0 \leq x$ is a type invariant and we have $x' = x
+ 3$ we need to prove $0 \leq x \land x' = x + 3 \implies 0 \leq x'$. However,
by introducing $\mathsf{nat}$ as a subtype of $\Itype$, typing $+:\mathsf{nat}
\cross \mathsf{nat} \fun \mathsf{nat}$ and declaring
$\Dec([x,x'],\mathsf{nat})$, it would be possible for the typechecker to 
automatically discharge the invariant by deciding
whether the program is type-correct or not. The greatest gains along this line
would be the introduction of a type for functions (say $\Ftype$) as a subtype of $\Rtype$ and
a library of function-based constraints that would typecheck when their
arguments are functions. This would make unnecessary the introduction of a
number of invariance lemmas ensuring that a given binary relation is indeed a
function.
For instance, the update or override operator available in B ($\lover$) and Z ($\oplus$) is known to be closed for functions. Hence, a type declaration such as $\lover: \Ftype(T,U) \cross \Ftype(T,U) \fun \Ftype(T,U)$ can be added to the system. In this case, many proof obligations that today are passed to a prover could be easily discharged by the typechecker, instead.
\bigskip

\noindent\textit{Competing interests: The authors declare none}

\bigskip

\bibliographystyle{acmtrans}
\bibliography{/home/mcristia/escritos/biblio}

\end{document}